\documentclass[11pt,a4paper,oneside,english]{article}
\usepackage{babel}
\usepackage[utf8]{inputenc}
\usepackage[T1]{fontenc} 
\usepackage{amsmath}
\usepackage{amsthm}
\usepackage{amssymb}
\usepackage{physics}
\usepackage{graphicx}
\usepackage{booktabs}
\usepackage[hidelinks]{hyperref}
\usepackage{makecell}
\usepackage{romannum}
\usepackage{longtable}
\usepackage[flushleft]{threeparttable}
\usepackage{fullpage,authblk}
\usepackage{float}
\usepackage{blindtext}
\usepackage{lmodern}
\usepackage[numbers]{natbib}
\usepackage{multirow}

\newtheorem{theorem}{Theorem}
\providecommand{\keywords}[1]{
  \small	
  \textbf{\textit{Keywords---}} #1
}

\title{\textbf{Efficient Mediated Multiparty Semi-Quantum Secret Sharing Protocol Based on Single-Qubit Reordering}}

\begin{document}

\author[1]{Mustapha Anis Younes \footnote{Corresponding Author: \texttt{mustaphaanis.younes@univ-bejaia.dz}}}
\author[2]{Sofia Zebboudj \footnote{\texttt{sofia.zebboudj@univ-ubs.fr}}}
\author[3]{Abdelhakim Gharbi \footnote{\texttt{abdelhakim.gharbi@univ-bejaia.dz}}}
\affil[1,2]{Université de Bejaia, Faculté des Sciences Exactes, Laboratoire de Physique Théorique, 06000 Bejaia, Algérie}
\affil[2]{ENSIBS, Université de Bretagne Sud, 56000 Vannes, France}
\date{\vspace{-8ex}}

\maketitle 
\pagenumbering{arabic}
\pagenumbering{arabic}

\begin{abstract}
    Typical multiparty semi-quantum secret sharing (MSQSS) protocols require the dealer to possess full quantum capabilities, while the classical users usually need to perform three operations. To address this practical limitation, this paper introduces a new mediated MSQSS protocol that enables Alice, a classical user, to share a secret with $M$ classical Bobs, with the assistance of an untrusted third party (TP) who may attempt any possible attack to steal Alice's secret. Furthermore, the classical participants require only two capabilities instead of three, namely: (a) performing measurements in the $Z$ basis; and (b) reordering qubits. The proposed scheme offers significant advantages over existing mediated QSS protocols: (1) it is the first mediated SQSS protocol to adopt single qubits, instead of entangled states, as the quantum resource, which makes it more practical and easier to implement; (2) It achieves higher qubit efficiency. Security analysis also demonstrates that the protocol is secure against well-known attacks.
\end{abstract}


\section{Introduction}

Secret sharing is a procedure that allows a dealer to share a secret among several participants. Invented independently by Shamir \cite{Shamir1979} and Blakley \cite{BLAKLEY1979} in 1979, it involves splitting the secret into multiple parts (called shadows) and distributing them to the participants. This is done in such a way that no individual part reveals any intelligible information about the original secret. Only when a sufficient number of participants combine their shadows can the secret be reconstructed. The security of classical secret sharing (CSS) protocols relies on computational complexity and hard mathematical problems, which makes them vulnerable to quantum computing attacks \cite{Chuang1995,Plenio1996}. On the other hand, quantum secret sharing (QSS) can overcome this challenge by relying on the fundamental laws of quantum physics.

In 1999, Hillery et al. \cite{Hillery1999} introduced the first QSS protocol based on GHZ states. Since then, numerous QSS protocols and experimental implementations have been proposed \cite{Zhang2005,Tyc2002,Bagherinezhad2003,Sarvepalli2012,Karlsson1999,Yu2008,Tavakoli2015}, leveraging the properties of various quantum resources. However, these protocols typically require participants to possess full quantum capabilities, which is difficult to achieve in practice \cite{Cirac2020}, as not everyone can afford expensive quantum devices.

To address this issue, Boyer et al. \cite{Boyer2007} introduced the concept of a "semi-quantum environment", which includes two types of users: quantum and classical. According to the definition, quantum users possess full quantum capabilities, whereas classical users are restricted to performing the following operations: (1) reflect particles without disturbance; (2) measure qubits in the $Z$ basis $\{\ket{0}, \ket{1}\}$; (3) prepare qubits in the $Z$ basis; and (4) reorder qubits. In 2007, Boyer et al. \cite{Boyer2007,Boyer2009} proposed the first semi-quantum key distribution (SQKD) protocol. Since then, various SQKD protocols have been developed \cite{Zou2009,Boyer2017,Tsai2018,Wang2011,Iqbal2020}, allowing a quantum user to share secret keys with a classical user. In 2015, Krawec \cite{Krawec2015} introduced the mediated model, which involves an untrusted, fully quantum third party (TP) acting as a mediator to help two classical participants establish a secure key, further reducing the quantum burden on the users.

In 2011, Li et al. \cite{Li2010} proposed the first semi-quantum secret sharing (SQSS) protocol, in which a quantum dealer can share secret information with two classical participants using GHZ-type states. Following this, various SQSS protocols have been proposed \cite{Xie2015,Li2013,CHUNWEI2013,Tian2021, Tian2023,Hou2024,Xing2023,Gao2016,Xin2024,Ma2024,Younes2024}, many of which can accommodate multiple participants. Although these protocols are lighter than fully quantum ones, they share a major restriction: the dealer must always be the quantum user. It is therefore interesting, from both theoretical and experimental perspectives, to explore whether a classical user can assume the role of the dealer.

This restriction was first tackled by Tsai et al. \cite{Tsai2022} in 2021. Their approach, based on the mediated model and leveraging the properties of GHZ states, introduced the first mediated multiparty quantum secret sharing (MQSS) protocol. It enables a classical user to securely share a secret with other classical users, with the assistance of an adversarial, fully quantum third party (TP). However, this protocol suffers from extremely low qubit efficiency. In 2023, Tsai et al. \cite{Tsai2023} proposed another mediated MQSS protocol based on a measurement property of graph states. This new protocol achieves a qubit efficiency that is $2^{M-1}$ times higher than the first scheme, where $M$ is the number of participants. Although both protocols place all participants on equal footing in terms of capabilities, the TP still requires heavy quantum resources. In practice, the cost and complexity of generating and maintaining such entangled states remain prohibitively high. Ideally, protocols where both the TP and the participants require only minimal quantum capabilities, such as handling single-qubit states, would be far more practical.

To reduce TP's quantum burden, this paper introduces the first mediated multiparty semi-quantum secret sharing protocol (MSQSS) based on single qubits. In the proposed scheme, TP is only required to: (1) generate qubits in the state $\ket{+}$, and (2) measure qubits in the $Z = \{\ket{0}, \ket{1}\}$ and $X = \{\ket{+} = (\ket{0} + \ket{1})/\sqrt{2}, \ket{-} = (\ket{0} - \ket{1})/\sqrt{2}\}$ bases. As for the classical participants, they only need two capabilities: (a) measuring qubits in the $Z$ basis, and (b) reordering qubits. As a result, our protocol is more practical for real-world implementation. Furthermore, the use of the qubit reordering operation minimizes the number of discarded particles in the protocol, resulting in higher qubit efficiency compared to previous mediated QSS protocols. Our protocol also adopts a circular qubit transmission method, making it more scalable than tree-based methods, especially in multiparty scenarios. Finally, security analyses shows that the protocol is secure and can resist well-known attacks, such as the intercept-resend attack, fake-state attack, entanglement-measure attack, Trojan horse attacks, and collusion attacks.

The remainder of this paper is organized as follows. Section \ref{sec_protocol} describes the proposed protocol in detail. Section \ref{sec_example} provides a concrete example of the protocol. The security analysis is presented in Section \ref{sec_security}. Section \ref{sec_comp} discusses the efficiency analysis and provides a comparison with other similar schemes. Finally, a conclusion is given in section \ref{sec_conc}.

\section{The proposed protocol}\label{sec_protocol}

In this paper, we propose a new mediated semi-quantum QSS scheme, where Alice, a classical entity wants to share a secret with $M$ classical Bobs with the help of an untrusted third-party (TP) who might attempt any possible attack to steal Alice's secret without being detected. Namely, Alice is capable to perform the following operations:

\begin{enumerate}
    \item Generate and measure qubits in the $Z= \{\ket{0}, \ket{1}\}$ basis.
    \item Reorder qubits via different delay lines.
\end{enumerate}
On the other hand, the classical Bobs are only capable of performing two operations, namely:
\begin{itemize}
    \item Measuring qubits in the $Z$ basis.
    \item Reordering qubits.
\end{itemize}

As for TP, he is only required to perform the following operations:

\begin{enumerate}
    \item Generate qubits in the state $\ket{+}=\frac{1}{\sqrt{2}}(\ket{0}+\ket{1})$. 
    \item Measure qubits in the $X$ and $Z$ basis, such as 

    \begin{equation}
        X = \{\ket{+}=\frac{1}{\sqrt{2}}(\ket{0}+\ket{1}), \ket{-}=\frac{1}{\sqrt{2}}(\ket{0}-\ket{1})\}
    \end{equation}
\end{enumerate}

The proposed protocol adopts a circular qubit transmission method. Additionally, there exist a public authenticated classical channel between the participants. Since the TP is considered adversarial, the classical channel between him and the participants do not necessarily need to be authenticated.

\subsection{Steps of the protocols}

Let $S$ be a classical bit string that Alice wants to share, with $L$ its length. The procedure of our proposed MSQSS protocol unfolds as follows:

\subsubsection{Step 01: (Preparation by TP)}

TP generates a sequence $S_{TP}$ of $N = 4L(1+M\epsilon)$ qubits, where:

\begin{itemize}
    \item $L$ is the desired length of the final secret key.
    \item $M$ is the number of participants.
    \item $\epsilon$ is a parameter that satisfies $\epsilon < 1$.
\end{itemize}

Each qubit in the sequence is prepared in the state $\ket{+}$, and the sequence is then transmitted to Alice.

\subsubsection{Step 02: (Alice's operations)}

Upon receiving the sequence $S_{TP}$, Alice randomly selects half of the qubits to measure in the $Z$ basis. She then replaces those qubits with newly generated ones in the same state she found. For convenience, we refer to those qubits as \textit{SIFT} particles, and the remaining ones as \textit{CTRL} particles. Alice ends up with a new sequence, denoted as $S_A$. Before transmitting the sequence to $\text{Bob}_1$, Alice reorders randomly the $N$ qubits. Note that the specific rearrangement order of the qubits is only known to Alice. The resulting sequence, denoted as $S_A^\prime$, is then sent to $\text{Bob}_1$.

%Upon receiving the sequence $S_{TP}$, Alice randomly selects half of the qubits to measure in the $Z$ basis. She then replaces those qubits with newly generated ones in the same state she found. For convenience, we refer to those qubits as \textit{SIFT} particles, and the remaining ones as \textit{CTRL} particles. Alice ends up with a new sequence, denoted as $S_A$, in which the \textit{SIFT} particles are prepared in the $Z$ basis, and the \textit{CTRL} particles are all in the state $\ket{+}$. Before transmitting the sequence to $\text{Bob}_1$, Alice reorders randomly the $4N(1+M\epsilon)$ qubits. Note that the specific rearrangement order of the qubits is only known to Alice. The resulting sequence, denoted as $S_A^\prime$, is then sent to $\text{Bob}_1$.

\subsubsection{Step 03: (Participants' operations)}

After receiving the sequence $S_A^\prime$ from Alice, $\text{Bob}_1$ randomly selects a fraction of size $4N\epsilon$ of the sequence and performs a $Z$ basis measurement on those qubits. He saves the positions of the measured qubits along with the corresponding measurement outcomes in his classical register. Next, $\text{Bob}_1$ randomly reorders the remaining qubits, creating a new sequence $S_{B_1}$, which he sends to $\text{Bob}_2$. Upon receiving the sequence, $\text{Bob}_2$ performs the same operations as $\text{Bob}_1$ and sends his resulting sequence to the next Bob. Each subsequent Bob follows the same procedure, except for the last one who sends his sequence $S_{B_M}$ back to TP. 

Note that when $N$ is large enough, it is sufficient for each Bob to select a subset approximately half the length of the secret to estimate the error rate with Alice. Therefore, setting $\epsilon = 1/8$ or less, for example, is a choice that is both valid and reasonable.

\subsubsection{Step 04: (TP's operations)} 

When TP receives the sequence $S_{B_M}$ from $\text{Bob}_M$, he randomly chooses to measure each qubit in either the $X$ basis or the $Z$ basis. After performing the measurements, he publicly announces both the measurement basis and the corresponding outcome for each qubit.

\subsubsection{Step 05: (Eavesdropping check)}

In this step, Alice conducts an eavesdropping check with the $M$ participants using the authenticated public channel. Alice requests all classical participants to reveal the positions of the qubits they measured, along with their corresponding measurement outcomes as follows:

\begin{itemize}
    \item $\text{Bob}_1$ announces the positions of the qubits he measured as well as the measurement outcome for each qubits. 
    \item For the remaining participants ($\text{Bob}_i$, $i\geq 2$), the information must be revealed in the following manner:
    \begin{enumerate}
    \item $\text{Bob}_i$ first reveals the positions of the measured qubits.
    \item The previous participants, starting from $\text{Bob}_{i-1}$ and moving in reverse order, disclose the transposition order of these announced qubits, so Alice can accurately perform her security check.
    \item After that, $\text{Bob}_i$ reveals the corresponding measurement outcome for each qubit. 
\end{enumerate}
\end{itemize}

For each participant, Alice compares the announced measurement outcomes with the corresponding states in her original sequence.

\begin{itemize}
    \item For the \textit{SIFT} particles, $\text{Bob}_i$'s  ($i\geq 1$) outcomes must match the states that Alice found in step 02.
    \item For \textit{CTRL} particles, Alice verifies if $\text{Bob}_i$'s outcomes are evenly distributed between $\ket{0}$ and $\ket{1}$.
\end{itemize}

If the error rate for the \textit{SIFT} qubits or the deviation rate for the \textit{CTRL} qubits exceeds a preset threshold, Alice aborts the protocol.

\subsubsection{Step 06: (TP's honesty check)}

In this step, Alice verifies the honesty of TP. She begins by requesting $\text{Bob}_M$ to reveal the transposition order of the qubits that TP measured in the $X$ basis. Each remaining participant, starting from the last and proceeding in reverse order, then publicly announces the transposition order of these qubits. Once this is over, Alice publicly announces the position of the \textit{CTRL} particles in her sequence $S_A^\prime$. Following this, $\text{Bob}_1$ reveals the transposition order of these qubits, and each remaining participant, in turn, publicly announces the transposition order for these qubits.

It is important to emphasize that Alice does not reveal her measurement outcomes for the \textit{SIFT} qubits, nor does she disclose the correct reordering for these qubits.

Depending on the operation performed by TP and the specific qubit involved, four equally likely cases arise:

\begin{enumerate}
    \item \textbf{Case 01: ($X$-\textit{CTRL})} TP performed an $X$ basis measurement on a \textit{CTRL} qubit. This case is used for eavesdropping detection. For TP to pass the check, he must consistently announce the measurement result $\ket{+}$; otherwise, Alice and the participants will abort the protocol.

    \item \textbf{Case 02: ($X$-\textit{SIFT})} TP performed an $X$ basis measurement on a \textit{SIFT} qubit. In this case, Alice and the participants expect TP to announce both the measurement results $\ket{+}$ and $\ket{-}$ with equal probability. Alice checks whether the results are evenly distributed, and if the deviation rate exceeds a predetermined threshold, she aborts the protocol.

    \item \textbf{Case 03: ($Z$-\textit{CTRL})} TP performed a $Z$ basis measurement on a \textit{CTRL} qubit. In this case, Alice and the participants expect TP to announce both the measurement results $\ket{0}$ and $\ket{1}$ with equal probability. Alice verifies whether the results are evenly distributed, and if the deviation exceeds a predetermined threshold, she aborts the protocol.

    \item \textbf{Case 04: ($Z$-\textit{SIFT})} TP performed a $Z$ basis measurement on a \textit{SIFT} qubit. The qubits in this case are used to establish the secret sharing key. The measurement outcomes announced by TP correspond to Alice's original measurement outcomes, that we denoted as $K_i^A$, but are randomly shuffled due to the reordering performed by Alice and the participants. Alice's secret key can only be reconstructed if Alice and all classical participants cooperate by sharing their transposition order of these qubits.

    To ensure security, Alice randomly selects a few of these bits and reveals their positions. $\text{Bob}_1$ then discloses the transposition order of these bits, followed by each remaining participant. Alice compares TP's outcomes with her own and if they do not align, the protocol is aborted and restarted. It is important to note that Alice only reveals the positions of those bits and not the outcomes.
\end{enumerate}

\subsubsection{Step 08: (Secret sharing)}

Once Alice confirms the absence of any eavesdropping or dishonest behavior from TP, she discloses her rearrangement order for the \textit{SIFT} qubits belonging to \textit{Case 04}. However, she does not reveal the corresponding measurement outcomes.

Alice now possess a random bit string $K^A$ of length $L$, which serves her as a secret key. She then uses $K_i$ to encrypt her secret bit string $S$ as follows:

\begin{equation}
    C_i = S_i \oplus K_i^A, \quad \text{for each $i$} \in \{1, 2, \cdots, N\}
\end{equation}

Alice then announces $C$ through the authenticated classical channel to share her secret information. When all classical Bobs cooperate by sharing their rearrangement orders, they can reconstruct Alice's secret key $K^A$ from TP's measurement outcomes and decrypt $C$ to retrieve her final secret.

\subsection{An example}\label{sec_example}

Now, we present an example of the proposed multiparty MSQSS protocol, where Alice intends to share a secret of length $N=5$ with two participants, Bob and Charlie. We take $\epsilon=\frac{1}{6}$. In this example, we suppose that TP and the participants are honest.

\subsubsection{TP's preparation}

Suppose TP prepares a sequence of $26$ qubits, each in the state $\ket{+}$ and sends it Alice.

\subsubsection{Alice's operations}

Upon receiving the sequence, Alice applies the a $Z$ basis measurement to the qubits in positions $(1, 2, 3, 6, 7, 10, 13, 15, 16, 18, 19, 23, 24, 25, 26)$. After her operations, she can end up with the following sequence:

\begin{align}
    S_{A} &= \{\ket{0}_1, \ket{1}_2, \ket{1}_3, \ket{+}_4, \ket{+}_5, \ket{0}_6, \ket{1}_7, \ket{+}_8, \ket{+}_9, \ket{0}_{10}, \ket{+}_{11}, \ket{+}_{12}, \nonumber\\ 
    & \hspace{1cm} \ket{0}_{13}, \ket{+}_{14}, \ket{1}_{15}, \ket{0}_{16}, \ket{+}_{17}, \ket{0}_{18}, \ket{0}_{19}, \ket{+}_{20}, \ket{+}_{21}, \ket{+}_{22}, \ket{1}_{23}, \ket{1}_{24}, \ket{0}_{25}, \ket{1}_{26}\}.
\end{align}
She shuffles this sequence by the following table:

\begin{align}\label{tab_alice}
\left(
\begin{array}{*{26}{c}}
    1 & 2 & 3 & 4 & 5 & 6 & 7 & 8 & 9 & 10 \\ 11 & 12 & 13 & 14 & 15 & 16 & 17 & 18 & 19 & 20 \\
    21 & 22 & 23 & 24 & 25 & 26
\end{array}
\right)
\longrightarrow
\left(
\begin{array}{*{26}{c}}
    22 & 19 & 15 & 5 & 12 & 14 & 3 & 21 & 4 & 20 \\
    16 & 6 & 11 & 1 & 8 & 7 & 9 & 13 & 23 & 17 \\
    18 & 24 & 10 & 2 & 26 & 25
\end{array}
\right)
\end{align}

This should read as: first qubit was displaced to position 22, second qubit to position 19, third qubit to position 15 $\cdots$ etc. Therefore, the sequence of qubits turns into:

\begin{align}
    S_{A}^\prime &= \{\ket{+}_{14}, \ket{1}_{24}, \ket{1}_{7}, \ket{+}_{9}, \ket{+}_{4}, \ket{+}_{12}, \ket{0}_{16}, \ket{1}_{15}, \ket{+}_{17}, \ket{1}_{23}, \ket{0}_{13}, \ket{+}_{5}, \nonumber\\ 
    & \hspace{1cm} \ket{0}_{18}, \ket{0}_{6}, \ket{1}_{3}, \ket{+}_{11}, \ket{+}_{20}, \ket{+}_{21}, \ket{1}_{2}, \ket{0}_{10}, \ket{+}_{8}, \ket{0}_{1}, \ket{0}_{19}, \ket{+}_{22}, \ket{1}_{26}, \ket{0}_{25}\}.
\end{align}
The subscript in each vector refers to the initial position of the qubit in Alice's original sequence $S_A$. In the rest of the example, we keep those subscripts to keep track of how Alice's original qubits are being reordered.

Alice sends the sequence $S_A^\prime$ to Bob.

\subsubsection{Bob's operations}

When Bob receives $S_A^\prime$, he randomly selects the qubits at the position 3, 25, and 26 to measure in the $Z$ basis. Therefore, he obtains:

\begin{align}
    \ket{1}_7 &\stackrel{Measure}{\longrightarrow} \ket{1}_7, \\
    \ket{1}_{26} &\stackrel{Measure}{\longrightarrow} \ket{1}_{26},\\
    \ket{0}_{25} &\stackrel{Measure}{\longrightarrow} \ket{0}_{25}.
\end{align}
After discarding those qubits from $S_A^\prime$, Bob ends up with the following sequence:

\begin{align}
    S_{B} &= \{\ket{+}_{14}, \ket{1}_{24}, \ket{+}_{9}, \ket{+}_{4}, \ket{+}_{12}, \ket{0}_{16}, \ket{1}_{15}, \ket{+}_{17}, \ket{1}_{23}, \ket{0}_{13}, \ket{+}_{5}, \nonumber\\ 
    & \hspace{1cm} \ket{0}_{18}, \ket{0}_{6}, \ket{1}_{3}, \ket{+}_{11}, \ket{+}_{20}, \ket{+}_{21}, \ket{1}_{2}, \ket{0}_{10}, \ket{+}_{8}, \ket{0}_{1}, \ket{0}_{19}, \ket{+}_{22}\}.
\end{align}
Bob shuffles this sequence by the order:

\begin{align}\label{tab_bob}
\left(
\begin{array}{*{23}{c}}
    1 & 2 & 3 & 4 & 5 & 6 & 7 & 8 & 9 & 10 \\  11 & 12 & 13 & 14 & 15 & 16 & 17 & 18 & 19 & 20 \\
    21 & 22 & 23
\end{array}
\right)
\longrightarrow
\left(
\begin{array}{*{23}{c}}
    13 & 17 & 15 & 8 & 4 & 5 & 21 & 14 & 12 & 19  \\
    3 & 2 & 18 & 23 & 10 & 6 & 16 & 9 & 11 & 1 \\
    22 & 7 & 20
\end{array}
\right)
\end{align}
This should read as: first qubit in $S_B$ was displaced to position 13, second qubit to position 17 $\cdots$ etc. Therefore, the sequence of qubits turns into:

\begin{align}
    S_{B}^\prime &= \{\ket{+}_8, \ket{0}_{18}, \ket{+}_5, \ket{+}_{12}, \ket{0}_{16}, \ket{+}_{20}, \ket{0}_{19}, \ket{+}_4, \ket{1}_2, \ket{+}_{11}, \ket{0}_{10}, \ket{1}_{23}, \nonumber\\ 
    & \hspace{1cm} \ket{+}_{14}, \ket{+}_{17}, \ket{+}_{9}, \ket{+}_{21}, \ket{1}_{24}, \ket{0}_{6}, \ket{0}_{13}, \ket{+}_{22}, \ket{1}_{15}, \ket{0}_{1}, \ket{1}_{3}\}.
\end{align}
Bob proceeds to send this sequence to Charlie.

\subsubsection{Charlie's operations}

Charlie randomly selects the qubits at positions 2, 6, and 13 of $S_B^\prime$ to measure in the $Z$ basis, which means that Charlie measures the qubits $\ket{0}_{18}$, $\ket{+}_{20}$, and $\ket{+}_{14}$, respectively. Charlie can obtain the following outcomes:

\begin{align}
    \ket{0}_{18} &\stackrel{Measure}{\longrightarrow} \ket{0}_{18}, \\
    \ket{+}_{20} &\stackrel{Measure}{\longrightarrow} \ket{1}_{20},\\
    \ket{+}_{14} &\stackrel{Measure}{\longrightarrow} \ket{0}_{14}.
\end{align}
After discarding those qubits from $S_B^\prime$, she ends up with the following sequence:

\begin{align}
    S_{C} &= \{\ket{+}_8, \ket{+}_5, \ket{+}_{12}, \ket{0}_{16}, \ket{0}_{19}, \ket{+}_4, \ket{1}_2, \ket{+}_{11}, \ket{0}_{10}, \ket{1}_{23}, \nonumber\\ 
    & \hspace{1cm}  \ket{+}_{17}, \ket{+}_{9}, \ket{+}_{21}, \ket{1}_{24}, \ket{0}_{6}, \ket{0}_{13}, \ket{+}_{22}, \ket{1}_{15}, \ket{0}_{1}, \ket{1}_{3}\},
\end{align}
which she shuffles by the following order:

\begin{align}\label{tab_char}
\left(
\begin{array}{*{20}{c}}
    1 & 2 & 3 & 4 & 5 & 6 & 7 & 8 & 9 & 10 \\  11 & 12 & 13 & 14 & 15 & 16 & 17 & 18 & 19 & 20
\end{array}
\right)
\longrightarrow
\left(
\begin{array}{*{20}{c}}
    7 & 5 & 1 & 20 & 8 & 18 & 12 & 6 & 15 & 9 \\  13 & 10 & 11 & 17 & 16 & 4 & 3 & 19 & 2 & 14
\end{array}
\right)
\end{align}

Therefore, Charlie ends up with the sequence $S_C^\prime$ as follow:

\begin{align}
    S_{C}^\prime &= \{\ket{+}_{12}, \ket{0}_{1}, \ket{+}_{22}, \ket{0}_{13}, \ket{+}_{5}, \ket{+}_{11}, \ket{+}_{8}, \ket{0}_{19}, \ket{1}_{23}, \ket{+}_{9}, \nonumber\\ 
    & \hspace{1cm} \ket{+}_{21}, \ket{1}_{2}, \ket{+}_{17}, \ket{1}_{3}, \ket{0}_{10}, \ket{0}_{6}, \ket{1}_{24}, \ket{+}_{4}, \ket{1}_{15}, \ket{0}_{16}\}.
\end{align}
Charlie sends this sequence to TP.

\subsubsection{TP's operations}

Upon receiving $S_C^\prime$, TP performs the $X$ basis measurement in the positions $(1, 2, 3, 5, 8, 10, 14, 17, 18)$ and the $Z$ basis measurement in the remaining positions. Then TP announces his measurement basis for each qubit as well as his measurement outcomes. If TP is honest, he announces the following sequence: 

\begin{align}
    S^\prime_{TP} &= \{\ket{+}_{12}, \ket{-}_{1}, \ket{+}_{22}, \ket{0}_{13}, \ket{+}_{5}, \ket{1}_{11}, \ket{0}_{8}, \ket{+}_{19}, \ket{1}_{23}, \ket{+}_{9}, \nonumber\\ 
    & \hspace{1cm} \ket{1}_{21}, \ket{1}_{2}, \ket{0}_{17}, \ket{+}_{3}, \ket{0}_{10}, \ket{0}_{6}, \ket{-}_{24}, \ket{+}_{4}, \ket{1}_{15}, \ket{0}_{16}\}.
\end{align}

\subsubsection{Eavesdropping check}

Following Alice’s request, Bob and Charlie announces the positions and outcomes of their selected qubits as indicated in the steps of the protocol. Alice verifies if their outcomes on the \textit{SIFT} qubits are aligned with her own, which is the case in the following situation. and that Charlie’s outcomes on the \textit{CTRL} qubits are evenly distributed, which is also satisfied.

It is important to note that Bob must take into account his discarded qubits on $S_B$ to determine the correct corresponding positions of Charlie's selected qubits.

Based on this verification, Bob and Charlie pass the public discussion.

\subsubsection{TP's honesty check}

By using their respective transposition tables (Tables \ref{tab_bob} and \ref{tab_char}), Bob and Charlie publicly announce, in reverse order, the correct transposition of the qubits that TP measured in the $X$ basis. Note that they must account for the discarded qubits in order to determine the correct reordering. Then, after Alice announces the positions of the \textit{CTRL} particles, Bob and Charlie publish their transposition orders for the qubits at those corresponding positions. Based on the operations of Alice and TP, we obtain the following four subsequences of $S_{TP}^\prime$:

\begin{enumerate}
    \item Subset where TP measured the \textit{CTRL} qubits in the $X$ basis:
\begin{align}
    S_{CX} = \{\ket{+}_{12},\ket{+}_{22},\ket{+}_{5},\ket{+}_{9},\ket{+}_{4}\}.
\end{align}
As we can see, all TP's measurement results are the state $\ket{+}$ as they should be.

    \item Subset where TP measured the \textit{SIFT} qubits in the $X$ basis:
    \begin{align}
    S_{SX} = \{\ket{-}_{1},\ket{+}_{19},\ket{+}_{3},\ket{-}_{24}\}.
\end{align}
As we can see, we have an even distribution for the states $\ket{+}$ and $\ket{-}$. 

    \item Subset where TP measured the \textit{CTRL} qubits in the $Z$ basis:
    \begin{align}
    S_{CZ} = \{\ket{1}_{11},\ket{0}_{8},\ket{1}_{21},\ket{0}_{17}\}.
\end{align}
As we can see, we have an even distribution for the states $\ket{0}$ and $\ket{1}$.

    \item Subset where TP measured the \textit{SIFT} qubits in the $Z$ basis:
    \begin{align}
    S_{SZ} =\{\ket{0}_{13},\ket{1}_{23},\ket{1}_{2},\ket{0}_{10},\ket{0}_{6},\ket{1}_{15},\ket{0}_{16}\}.
\end{align}
Alice uses the measurement outcomes of the $2^{nd}$ and $6^{th}$ qubits of her sequence $S_A$ as test bits, which correspond to the $3^{rd}$ and $5^{th}$ qubits of $S_{SZ}$. As we can see, TP announced the correct measurement results. 
\end{enumerate}

After the eavesdropping detection passes, TP, Bob, and Charlie have access to the following result $('01010')$. By sharing their own rearrangement order, Bob and Charlie can obtain Alice's bit string $K= '00101'$.

\section{Security analysis}\label{sec_security}

In this section, we examine the security of the proposed protocol. Given that TP possesses greater capabilities than any external or internal eavesdropper, we focus on the scenario where TP acts as the primary adversary. To steal Alice's secret, TP might employ various attacks. Our analysis demonstrates that the protocol remains secure against well-known strategies, including the fake states attack, intercept-resend attack, entanglement-measure attack, and Trojan horse attacks. Additionally, we address the potential threat of collusion, where TP collaborates with one or more participants to compromise Alice's secret.

\subsection{Fake states attack}

In this attack, TP prepares qubits in states other than $\ket{+}$, as required of him in step 1 of the protocol. Suppose TP prepares his sequence in the $Z$ basis $\{\ket{0}, \ket{1}\}$. Although this strategy allows him to control the content of Alice's secret key, he cannot distinguish the \textit{SIFT} particles from the \textit{CTRL} particles. As a result, although TP’s attack goes undetected during the initial eavesdropping check between Alice and the participants, it will be detected with probability $1 - (7/8)^L$ during the honesty check. This probability converges to 1 as $L$ becomes sufficiently large.

\subsection{Intercept-resend attack}

In this attack, TP attempts to learn the rearrangement order of each participant. To do so, he intercepts the sequence $S_A^\prime$ sent by Alice to $\text{Bob}_1$ and stores it in his quantum memory. He then sends a fake sequence of particles to each $\text{Bob}_i$ in turn. After a participant completes his operations, TP intercepts the fake sequence and measures it in an attempt to deduce the participant’s rearrangement order. Once TP has retrieved and measured all the fake sequences, he applies the inferred rearrangement orders to the stored sequence $S_A^\prime$ and sends it back to Alice. However, this attack is bound to fail. To demonstrate this, we examine the following two strategies:

\begin{itemize}
    \item \textit{The fake sequences are composed of qudits:} without loss of generality, suppose that TP prepares the qubits in his fake sequences in the $Z$ basis. In this scenario, TP cannot distinguish whether two qubits in the same state but at different positions have been exchanged. For example, consider the sequence $\{\ket{1}, \ket{1}, \ket{0}, \ket{0}\}$. If it is reordered to $\{\ket{1}, \ket{0}, \ket{1}, \ket{0}\}$, TP cannot tell whether the first $\ket{1}$ remained in place or was moved to the third position. When the sequence is large, the probability of correctly guessing the transposition order becomes negligible. Furthermore, this attack would inevitably be detected during the eavesdropping check between Alice and the participants, since TP’s fake sequences do not match Alice’s original sequence $S_A^\prime$. As a result, the participants’ measurement outcomes in step 3 of the protocol would not necessarily align with the outcomes Alice recorded in step 2. Therefore, TP cannot gain any useful information using this strategy.

    \item \textit{The fake sequences are composed of qubits:} in this strategy, TP uses $n$-level quantum states to prepare his fake sequences. To deduce $\text{Bob}_1$'s rearrangement order, for instance, TP prepares the sequence $\{\ket{0}_N, \ket{1}_N, \ket{2}_N, \cdots, \ket{N-1}_N\}$, where the subscript $N$ is to denote that they are $N$ dimensional vectors. After $\text{Bob}_1$ completes his operations, TP retrieves the sequence and measures the particles in the $Z^{(N)}=\{\ket{0}_N, \ket{1}_N, \cdots, \ket{N-1}_N\}$ basis. Since all the states in the sequence are mutually orthogonal, TP can determine $\text{Bob}_1$'s rearrangement order with certainty. If $\text{Bob}_1$'s actions consisted solely of reordering the particles, TP's strategy would succeed. However, $\text{Bob}_1$ also randomly selects a subset of particles to measure in the two-level $Z$ basis $\{\ket{0}, \ket{1}\}$. In that case, TP’s attack is bound to be detected since $\text{Bob}_1$'s outcomes won't be aligned with Alice's results in step 02 of the protocol. The same reasoning applies to the other Bobs, meaning that TP cannot obtain any useful information using this strategy without being detected.

\end{itemize}

\subsection{Entanglement-measure attack}

%The entanglement-measure attack of TP is modeled by the unitary operations $(U_F, U_R)$, where $U_F$ is used to attack the particles sent by Alice to $\text{Bob}_1$, and $U_R$ is used to attack the particles sent by $\text{Bob}_M$ to TP. Let $\ket{a}_i$ be the $i-th$ qubit of the sequence sent by Alice, to this qubit, TP associates an ancillary particle initially in the state $\ket{f}_i$ to execute his attack $U_F$. The probability that $\ket{f}_i$ still matches $\ket{b}_i$ after the reordering of the participants is close to zero. Therefore, TP uses a new ancilla $\ket{g}_i$ during his second attack $U_R$.  

The entanglement-measure attack of TP is modeled by the unitary operations $(U_F, U_R)$, where $U_F$ is used to attack the particles sent by Alice to $\text{Bob}_1$, and $U_R$ is used to attack the particles sent by $\text{Bob}_M$ to TP. Note that TP does not attack the particles he sends to Alice, as they do not carry any information about Alice’s secret or the participants’ secret shadows. Furthermore, TP uses a different auxiliary probe, denoted as $F$ and $R$ respectively, for each unitary operation. The reason is that, since the sequence sent by Alice is reordered by the participants, TP cannot determine which ancilla $F$ corresponds to each retrieved particle.

\begin{theorem}
    Suppose TP performs an attack $(U_F, U_R)$ on the particles traveling from Alice to $\text{Bob}_1$ and from $\text{Bob}_M$ back to him, where $F$ and $R$ are TP’s auxiliary probes. TP introduces no error during the eavesdropping check if and only if the final states of his probes are independent of his measurement results on the qubits received by $\text{Bob}_M$. As a result, TP gains no information about Alice’s shared secret.
\end{theorem}

\begin{proof}

TP intercepts $S_A^\prime$ from Alice and applies $U_F$ to each traveling qubit along with its associated ancilla $F$, initially prepared in some arbitrary normalized state $\ket{f}$. The composite system then evolves as follows:

\begin{align}
    U_F (\ket{0}\ket{f}) &= \alpha \ket{0}\ket{f_0} + \beta \ket{1}\ket{f_1}, \\
    U_F (\ket{1}\ket{f}) &= \beta \ket{0}\ket{f_2} + \alpha \ket{1}\ket{f_3},
\end{align}
such as $\abs{\alpha}^2 + \abs{\beta}^2 = 1$, and $\ket{f_i}$ are states that TP can distinguish. By linearity, we obtain

\begin{align}
    U_F (\ket{+}\ket{f}) &= \frac{1}{\sqrt{2}} \ket{0} \big(\alpha \ket{f_0} + \beta \ket{f_2} \big) \\
                         &\quad\quad + \frac{1}{\sqrt{2}} \ket{1} \big(\beta \ket{f_1} + \alpha \ket{f_3} \big)
\end{align}

In order for TP to pass the eavesdropping check between Alice and the participants in \textbf{Step 05} of the protocol, he must adjust $U_F$ accordingly. Specifically, TP must set $\beta = 0$ to ensure that the participants do not obtain invalid outcomes, and $\alpha = 1$ to satisfy the normalization condition. Under these constraints, the above equations reduce to:

\begin{align}
    U_F (\ket{0}\ket{f}) &= \ket{0}\ket{f_0}, \\
    U_F (\ket{1}\ket{f}) &= \ket{1}\ket{f_3}, \\
    U_F (\ket{+}\ket{f}) &= \frac{1}{\sqrt{2}} \big(\ket{0}\ket{f_0} + \ket{1}\ket{f_3} \big)
\end{align}

When $\text{Bob}_M$ sends his sequence, TP attaches to each qubit a new probe, initially prepared in some arbitrary normalized state $\ket{r}$. He then applies the operation $(U_R \otimes I_F)$ to the composite system, which evolves as follows:

\begin{align}
    (U_R \otimes I_F)(U_F \ket{0}\ket{f}) &= \gamma \ket{0}\ket{r_0}\ket{f_0} + \delta \ket{1}\ket{r_1}\ket{f_0}, \\
    (U_R \otimes I_F)(U_F \ket{1}\ket{f}) &= \delta \ket{0}\ket{r_2}\ket{f_3} + \gamma \ket{1}\ket{r_3}\ket{f_3}, \\
\end{align}
    where $\abs{\gamma}^2 + \abs{\delta}^2 = 1$, and $\ket{r_i}$ are states that TP can distinguish. By linearity, we obtain

\begin{align}
    (U_R \otimes I_F)(U_F \ket{+}\ket{f}) &= \frac{1}{\sqrt{2}} \ket{0} \big( \gamma \ket{r_0}\ket{f_0} + \delta\ket{r_2}\ket{f_3} \big) \\ 
                                        &\quad\quad + \frac{1}{\sqrt{2}} \ket{1} \big( \delta \ket{r_1}\ket{f_0} + \gamma \ket{r_3}\ket{f_3} \big)
\end{align}

If TP wants to pass the honesty check, then the final state of the first register must be identical to the original state sent by Alice. This implies that TP must satisfy the conditions $\delta = 0$ and $\gamma = 1$. Therefore, $U_R$ is now defined as follows:

\begin{align} \label{U_R_z_2/3}
    (U_R \otimes I_F)(U_F \ket{0}\ket{f}) &= \ket{0}\ket{r_0}\ket{f_0}, \\
    (U_R \otimes I_F)(U_F \ket{1}\ket{f}) &= \ket{1}\ket{r_3}\ket{f_3}, \\
\end{align}
By linearity

\begin{align} \label{U_R_+_2/3}
    (U_R \otimes I_F)(U_F \ket{+}\ket{f}) &= \frac{1}{2} \ket{+} \big( \ket{r_0}\ket{f_0} + \ket{r_3}\ket{f_3} \big) \\
                                        & \quad\quad + \frac{1}{2} \ket{-} \big( \ket{r_0}\ket{f_0} - \ket{r_3}\ket{f_3} \big)
\end{align}
To go undetected, TP must set the incorrect terms as a zero vector, specifically:

\begin{align}
    \ket{r_0}\ket{f_0} = \ket{r_3}\ket{f_3},
\end{align}
which means that

\begin{align} \label{rf_0 = rf_3}
    \ket{f_0} &= \ket{f_3} = \ket{F}, \\
    \ket{r_0} &= \ket{r_3} = \ket{R}.
\end{align}
After inserting Eq. (\ref{rf_0 = rf_3}) into Eqs. (\ref{U_R_z_2/3}-\ref{U_R_+_2/3}), we obtain

\begin{align} \label{U_R_final}
    \begin{cases}
        (U_R \otimes I_F)(U_F \ket{0}\ket{f}) = \ket{0}\ket{R}\ket{F}, \\
        (U_R \otimes I_F)(U_F \ket{1}\ket{f}) = \ket{1}\ket{R}\ket{F}, \\
        (U_R \otimes I_F)(U_F \ket{+}\ket{f}) = \ket{+}\ket{R}\ket{F}.
    \end{cases}
\end{align}

According to Eq. (\ref{U_R_final}), when TP remains undetected during both the eavesdropping check with the participants and the honesty check, not only can he not distinguish the final states of his ancillary probes, but these states are also always independent of the \textit{CTRL} and \textit{SIFT} particles. As a result, TP obtains no information about Alice’s shared secret key.

\end{proof}

\subsection{Trojan horse attack}

Quantum Trojan horse attacks are common implementation attacks, in which TP inserts invisible or delayed photons into the particles he transmits to Alice. After retrieving these Trojan horse photons, TP can measure them to extract information about Alice's operations. He can use the same strategy with the other participants to obtain their secret shadows. Fortunately, the participants can easily defend against this attack by using a photon number splitter (PNS) and a wavelength filter device (WF). Therefore, the proposed protocol is secure against quantum Trojan horse attacks.

\subsection{Collusion attack}

In this attack, TP collaborates with some dishonest participants to obtain Alice's secret. Without loss of generality, consider the extreme scenario where only $\text{Bob}_i$ (along with Alice of course) is honest. A first strategy is to guess $\text{Bob}_i$'s reordering. However, since the permutation is chosen randomly and independently of the other Bobs, the probability of guessing it correctly is $1/n!$, which becomes negligible as $n$ grows large.

A second strategy is for TP to still apply the attacks previously described, while the $M-1$ dishonest Bobs announce fake transposition orders in an attempt to avoid detection during the eavesdropping and honesty checks. However, this strategy will inevitably fail. First, TP and the dishonest Bobs do not know which qubits $\text{Bob}_i$ will choose to measure in the $Z$ basis, nor their states. Furthermore, they cannot distinguish the \textit{CTRL} qubits from the \textit{SIFT} qubits in Alice's sequence $S_A^\prime$, as Alice only discloses this information after all Bobs have revealed the transposition orders of the \textit{$X$-CTRL} qubits. As a result, the fake transposition orders will introduce errors during both the eavesdropping and honesty checks, causing the protocol to be aborted.

Overall, our protocol is secure against collusion attacks.

\section{Efficiency analysis and comparison}\label{sec_comp}

In this section, we analyze the efficiency of the proposed protocol and compare its performance with existing multiparty mediated QSS protocols. The efficiency of our scheme, as well as those presented in Refs. \cite{Tsai2022,Tsai2023}, can be calculated using the following formula:

\begin{align}
    \eta = \frac{c}{q + b},
\end{align}
where $c$ is the length of the final secret key, $q$ is the number of qubits generated by TP, and $b$ is the number of qubits generated by the classical participants. In the proposed scheme, TP generates $q = 4L(1+M\epsilon)$ qubits. Alice measures approximately half of these qubits in the $Z$ basis and replaces them with newly generated ones in the same states she observed. Since she is the only classical participant who generates qubits, we have $b = 2L(1+M\epsilon)$. As for the $M$ Bobs, they select altogether $4LM\epsilon$ qubits to measure and forward the rest to TP. At the end, only the \textit{$Z$-SIFT} qubits are used as the final secret key, thus $c = L$. The qubit efficiency becomes:

\begin{align}
    \eta = \frac{1}{6+M\epsilon},
\end{align}

We now compare the proposed protocol with other multiparty mediated QSS protocols. The comparison is drawn from four perspective: quantum resources, quantum capabilities of the participants, communication structure, and qubit efficiency. The results are summarized in Table \ref{table 1}.

\begin{table}[h]
\caption{\label{table 1}Comparison of proposed protocol with other Mediated MQSS schemes.}
\resizebox{\textwidth}{!}{%
\begin{tabular}{lcccc}
\toprule
Protocol & \makecell{Quantum \\ resources} & \makecell{Capabilities of \\ classical participants} & \makecell{Communication \\ structure} & \makecell{Qubit \\ efficiency} \\
\midrule
Tsai et al. \cite{Tsai2022} & GHZ states & \multirow{2}{*}{\makecell[l]{1. Measure $\{\ket{0}, \ket{1}\}$\\2. Perform Hadamard $H$}} & \multirow{2}{*}{One-way} & $\cfrac{1}{2^{M+1}(M+1)}$ \\

Tsai et al. \cite{Tsai2023} & Graph states &  &  & $\cfrac{1}{4(M+1)}$ \\

Our protocol & Single qubits & \makecell[l]{1. Measure $\{\ket{0}, \ket{1}\}$ \\2. Reorder qubits} & Circular & $\cfrac{1}{6(1+M\epsilon)}$ \\
\bottomrule
\end{tabular}
}
\end{table}

In terms of quantum resources, the protocols of Tsai et al. \cite{Tsai2022, Tsai2023} require TP to generate multi-particle GHZ states and complete graph states, which are difficult to produce and maintain. In contrast, our protocol requires TP only to prepare $\ket{+}$ states and to perform single-qubit measurements in the $X$ and $Z$ bases. This significantly reduces TP’s quantum overhead, making the proposed protocol more practical and feasible in terms of implementation.

\begin{figure}[h]
    \centering
    \includegraphics[width=0.7\textwidth]{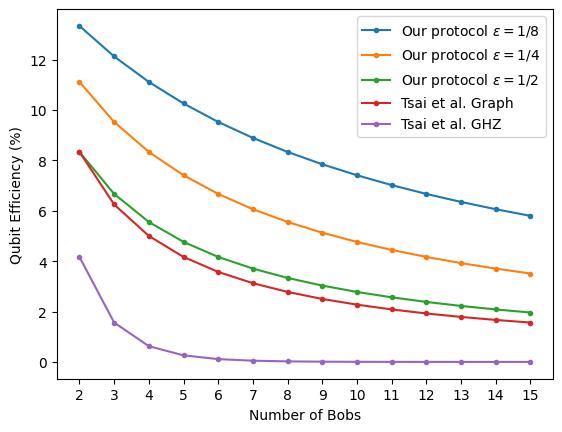}  
    \caption{Qubit efficiency of the different schemes for different numbers of Bobs.}
    \label{efficiency_fig}
\end{figure}

Regarding communication structure, our protocol uses a circular qubit transmission method, whereas Tsai et al.'s protocols adopt a one-way transmission method. This gives Tsai et al.'s protocols certain advantages, such as reduced qubit transmission distance. Additionally, it prevents the classical participants from needing additional devices to defend against quantum Trojan horse attacks. However, in terms of qubit efficiency, our protocol exhibits a significant advantage, especially over the protocol that is based on GHZ states. In Figure \ref{efficiency_fig}, we can more clearly compare the efficiency of our protocol for different values of $\epsilon$ with that of Tsai et al.'s protocols as a function of the number of classical Bobs. Specifically, for the protocol based on GHZ states, the qubit efficiency drops below 1\% (i.e. 0.625\%) when the number of classical Bobs reaches $4$, which makes the protocol extremely inefficient. As for the protocol based on graph states, the efficiency of our protocol remains superior regardless of the number of Bobs when $\epsilon \geq 1/2$. Note that setting $\epsilon$ to this value is an extreme choice. In an ideal case, and especially when $N$ is large, each Bob taking a small subset no larger than half the length of the secret is enough to evaluate the error rate with Alice. Therefore, setting $\epsilon = 1/8$ is a very reasonable choice.

Overall, even though Tsai et al.'s protocols shorten the qubit transmission distance and naturally ward off quantum Trojan horse attacks, our protocol demonstrates higher qubit efficiency and utilizes far cheaper quantum resources that are easier to handle, making it more feasible and efficient in terms of implementation. Furthermore, the circular communication structure gives our protocol an advantage in terms of scalability in the multiparty scenarios.

\section{Conclusion}\label{sec_conc}

This study introduces the first mediated MSQSS protocol based on single qubits. The protocol allows classical Alice to share a secret with $M$ classical Bobs. Compared to similar approaches, the quantum overhead of TP is significantly reduced, and the qubit efficiency is notably improved. As a result, the proposed scheme is more feasible and efficient in terms of implementation. It is also more practical than typical standard multiparty SQSS protocols, as (1) Alice does not need to possess full quantum capabilities, and (2) the classical participants only need to perform two operations (i.e. measuring in the $Z$ basis and reordering qubits). Furthermore, security analysis shows that the protocol can resist common attacks. In future work, it would be interesting to study the behavior of the protocol in the presence of noise.

\bibliographystyle{plainnat}
 \bibliography{mmsqss_reordering}

\end{document}